\documentclass[11pt]{article}

\usepackage{amsmath,amsfonts,amssymb,amsthm}
\usepackage[margin=1in]{geometry}
\usepackage[colorlinks]{hyperref}
\usepackage{enumitem}
\usepackage{mathtools}
\usepackage{enumitem}
\usepackage{algorithm, caption}
\usepackage[noend]{algpseudocode}
\usepackage{subfigure}
\usepackage{graphicx}
\usepackage{tabularx,environ}
\usepackage[braket]{qcircuit}
\usepackage[dvipsnames]{xcolor}
\usepackage{multirow}
\usepackage{tikz}
\usepackage{thm-restate}
\usepackage[capitalize]{cleveref}
\usepackage{dsfont}
\usepackage{enumitem}
\usepackage{comment}

\makeatletter
\newcommand*\bigcdot{\mathpalette\bigcdot@{.5}}
\newcommand*\bigcdot@[2]{\mathbin{\vcenter{\hbox{\scalebox{#2}{$\m@th#1\bullet$}}}}}
\makeatother

    \hypersetup{
    	colorlinks,
    	linkcolor={red!100!black},
    	citecolor={blue!100!black},
        urlcolor = {blue!60!black}
    }
\makeatletter
\newcommand\notsotiny{\@setfontsize\notsotiny\@vipt\@viipt}
\makeatother

\newcommand{\C}{{\mathbb{C}}}

\newcommand{\norm}[1]{\left\|{#1}\right\|}

\def\01{\{0,1\}}

\newcommand{\eps}{\varepsilon}
\newcommand{\ketbra}[2]{|#1\rangle\langle#2|}

\theoremstyle{plain}

\newtheorem{theorem}{Theorem}

\newtheorem{lemma}[theorem]{Lemma}

\newtheorem{fact}[theorem]{Fact}
\newtheorem{result}[theorem]{Result}


\definecolor{applegreen}{rgb}{0.0, 0.5, 0.0}

\renewcommand{\Pr}{\mbox{\rm Pr}}

\DeclareMathOperator{\op}{op} 

\newcommand{\N}{\mathbb{N}} 

\newcommand{\poly}{\mbox{\rm poly}}

\DeclareMathOperator{\Tr}{Tr}
\DeclareMathOperator{\tr}{tr}

\newcommand{\beq}{\begin{equation}}
\newcommand{\eeq}{\end{equation}}
\newcommand{\beqn}{\begin{equation*}}
\newcommand{\eeqn}{\end{equation*}}
\newcommand{\beqr}{\begin{eqnarray}}
\newcommand{\eeqr}{\end{eqnarray}}
\newcommand{\beqrn}{\begin{eqnarray*}}
\newcommand{\eeqrn}{\end{eqnarray*}}
\newcommand{\bmline}{\begin{multline}}
\newcommand{\emline}{\end{multline}}
\newcommand{\bmlinen}{\begin{multline*}}
\newcommand{\emlinen}{\end{multline*}}

\theoremstyle{plain}

\theoremstyle{definition}

\theoremstyle{remark}

\renewenvironment{proof}[1][]{
    \begin{trivlist}
    \item[\hspace{\labelsep}{\em\noindent Proof#1:\/}]}
    {{\hfill$\Box$}
    \end{trivlist}
}

\newtheoremstyle{named}{}{}{\itshape}{}{\bfseries}{.}{.5em}{\thmnote{#3}}
\theoremstyle{named}

\begin{document}

\title{\textbf{Certifying and learning local quantum Hamiltonians}}

\author{Andreas Bluhm\thanks{Univ. Grenoble Alpes, CNRS, Grenoble INP, LIG}\\{\small \href{mailto:Andreas.Bluhm@univ-grenoble-alpes.fr}{\texttt{Andreas.Bluhm@univ-grenoble-alpes.fr}}} \and 
Matthias C. Caro\thanks{University of Warwick}\\{\small \href{mailto:matthias.caro@warwick.ac.uk}{\texttt{matthias.caro@warwick.ac.uk}}} \vspace{0.8em}
\and
Francisco Escudero Gutiérrez\thanks{Centrum Wiskunde \& Informatica, QuSoft}\\{\small \href{mailto:Francisco.Escudero.Gutierrez@cwi.nl}{\texttt{Francisco.Escudero.Gutierrez@cwi.nl}}} \and
Junseo Lee\thanks{Institute of Computer Technology, Seoul National University}\\{\small \href{mailto:junseolee@fas.harvard.edu}{\texttt{junseolee@fas.harvard.edu}}}\vspace{0.8em}
\and 
Aadil Oufkir\thanks{University Mohammed VI Polytechnic}\\{\small \href{mailto:aadil.oufkir@gmail.com}{\texttt{aadil.oufkir@gmail.com}}} \and
Cambyse Rouzé\thanks{Inria, Télécom Paris - LTCI Institut Polytechnique de Paris}\\{\small \href{mailto:cambyse.rouze@telecom-paris.fr}{\texttt{cambyse.rouze@telecom-paris.fr}}} \and
Myeongjin Shin\thanks{Korea Advanced Institute of Science and Technology}\\{\small \href{mailto:hanwoolmj@kaist.ac.kr}{\texttt{hanwoolmj@kaist.ac.kr}}}}

\date{\today}

\maketitle

\begin{abstract}
In this work, we study the problems of certifying and learning quantum $k$-local Hamiltonians, for a constant $k$. Our main contributions are as follows:\\

\textbf{Certification of Hamiltonians.} We show that certifying a local Hamiltonian in normalized Frobenius norm via access to its time-evolution operator can be achieved with only $O(1/\varepsilon)$ evolution time. This is optimal, as it matches the Heisenberg-scaling lower bound of $\Omega(1/\varepsilon)$. To our knowledge, this is the first optimal algorithm for testing a Hamiltonian property. A key ingredient in our analysis is the Bonami Hypercontractivity Lemma from Fourier analysis.\\

\textbf{Learning Gibbs states.} We design an algorithm for learning Gibbs states of local Hamiltonians in trace norm that is sample-efficient in all relevant parameters. In contrast, previous approaches learned the underlying Hamiltonian (which implies learning the Gibbs state), and thus inevitably suffered from exponential sample complexity scaling in the inverse temperature.\\

\textbf{Certification of Gibbs states.} We give an algorithm for certifying Gibbs states of local Hamiltonians in trace norm that is both sample and time-efficient in all relevant parameters, thereby solving a question posed by Anshu (Harvard Data Science Review, 2022).
\end{abstract}

\newpage
\section{Introduction}
To make good on the promise of quantum computers and quantum simulators to solve problems intractable for classical computers, we will require quantum hardware composed of thousands of qubits \cite{bravyi2024high,liu2025robust}. However, the ultimate performance of this hardware will depend on the accurate implementation of carefully engineered dynamics, both to realize the desired computation and to mitigate errors. Therefore, the ability to precisely characterize the dynamics of quantum systems involving many qubits becomes crucial. This problem is known as Hamiltonian learning and has recently received a lot of attention due to its fundamental importance for quantum technologies \cite{Silva2011Practical,holzapfel2015scalable,Zubida2021Optimal,haah2022optimal,yu2023robust,Dutkiewicz.2023,huang2023heisenberg,li2023heisenberglimited,möbus2023dissipationenabled,franca2024efficient,Gu2022Practical,zhao2025learning,hu2025ansatz,anshu2021sample,rouze2023learning,onorati2023efficient,bakshi2023learning,ma2024learning,arunachalam2025testing,caro2023learning,castaneda2023hamiltonian}. More recently, Hamiltonian testing \cite{aharonov2022quantum,laborde2022quantum,she2022unitary,bluhm2024hamiltonian,arunachalam2025testing,kallaugher2025hamiltonian,gao2025quantum,sinha2025improvedhamiltonianlearningsparsity} has emerged as a potentially more efficient way to ensure that a quantum device works as desired.

    In fact, one of the main motivations for Hamiltonian learning comes from a property testing task: Given a quantum device of which a provider claims that it behaves according to a given Hamiltonian, how can we efficiently verify that behavior? 
    Of course, one may learn the Hamiltonian governing the system and then check if it coincides with the target Hamiltonian. However, full Hamiltonian learning is in general unnecessarily resource-intensive if we ``just'' want to certify that the Hamiltonian is indeed the one claimed by the provider. 
    This motivates a shift in focus from Hamiltonian learning to \emph{quantum Hamiltonian certification}, which, surprisingly, was only explored before in \cite{gao2025quantum}. 

    In \cite{gao2025quantum}, Gao, Ji, Wang, Yu, and Zhao proposed an efficient, but not provably optimal, certification protocol for sparse Hamiltonians by accessing the Hamiltonian via its dynamics. In contrast, in this work we focus on finding an optimal protocol for one of the most physically relevant classes of quantum Hamiltonians on $n$ quantum particles: the family of $k$-local Hamiltonians. We will think of $k$ as a constant, but we will make the dependence on $k$ explicit in our theorem statements. These Hamiltonians are those that can be written as a sum of Hamiltonian terms that act non-trivially on at most $k$ particles. Such $k$-local interactions model many relevant quantum systems, as they, for example, include the case of geometrically local Hamiltonians, where particles only interact with their nearest neighbors, and other models such as the SYK model \cite{sachdev1993gapless,Kitaev:2015a}.  Moreover, such Hamiltonians with $k$-particle interactions have played a prominent role in quantum complexity theory \cite{oliveira2005complexity, kempe2006complexity, bravyi2017complexity}, with the corresponding $k$-local Hamiltonian problem proven to be QMA-complete; and they are known to be universal for quantum simulation \cite{cubitt2016complexity, cubitt2018universal}. Finally, the special cases of $2$-local classical and quantum Hamiltonians have been extensively studied, since, despite their apparent simplicity, they exhibit non-trivial behavior. For the classical case, see \cite{mccoy2012importance,daskalakis2019testing} on the relevance of Ising models. For the quantum case, 2-local Hamiltonians exhibit non-trivial quantum phase transitions \cite{dutta2010quantum,suzuki2012quantum}, and include several physically relevant models \cite{lipkin1965validity,sachdev1993gapless,hormozi2017nonstoquastic}. Thus, the main motivating question of this work is: 
    \begin{quote}
    \center \emph{What is the complexity of certifying local Hamiltonians from dynamics?}
    \end{quote}\vspace{0.4cm}

    We answer this question by showing that $O(1/\eps)$ evolution time is sufficient for distinguishing whether an unknown Hamiltonian is equal to some desired target Hamiltonian or $\eps$-far from it. This is optimal, due to known lower bounds, and achieves the gold standard of Heisenberg-scaling, i.e. the complexity dependence on $\eps$ scales as $1/\eps$. Furthermore, to the best of our knowledge, this is the first provably optimal result in Hamiltonian property testing. Additionally, we show that our $k$-local Hamiltonian certification algorithm can be made tolerant. 
    
    In addition, we also explore the case when the Hamiltonian is accessed via its Gibbs, or thermal, state $\rho(\beta)=e^{-\beta H}/\Tr[e^{-\beta H}]$, where $\beta$ is the inverse temperature. In this setting, Anshu asked whether certification is possible in an efficient manner \cite{anshu2022some}. However, to the best of our knowledge, not even the more studied problem of learning Gibbs states of local Hamiltonians has an algorithm that is efficient in all relevant parameters \cite{anshu2021sample, haah2022optimal, rouze2023learning, onorati2023efficient, bakshi2023learning, Gu2022Practical,chen2025learning}. That is because previous approaches learn the Hamiltonian (which implies learning the Gibbs state, as noted in \cite[Remark 18]{anshu2021sample}), but due to known lower bounds, this requires an exponential dependence on $\beta$~\cite[Theorem 1.2]{haah2022optimal}. 

    For learning thermal states of local Hamiltonians, we give, to our knowledge, the first algorithm that is sample-efficient in all relevant parameters. Furthermore, for certifying such states, we give the first algorithm that is both sample- and time-efficient in all relevant parameters, resolving the question by Anshu \cite{anshu2022some}.

    \subsection{Results and Technical Overview}
    We will consider $n$-qubit Hamiltonians $H$, and we assume that they are $k$-local, for a constant $k\in \N$. As such, their expansion in the Pauli basis is simply
    \begin{equation*}
        H=\sum_{P\in \{I,X,Y,Z\}^{\otimes n}:\, |P|\leq k}h_P P\, ,
    \end{equation*}
    where $|P|$ is the number of sites where the Pauli string differs from identity. As two Hamiltonians that only differ in a multiple of the identity induce the same dynamics and thermal states, we assume without loss of generality that $h_{I^{\otimes n}}=\Tr[H]/2^n=0.$
    
    \subsubsection{Certification via access to the dynamics}
    If a quantum system is governed by a Hamiltonian $H$, then, according to the Schr\"odinger equation, its dynamics are determined by the unitary time evolution operator $U_H(t)=e^{-itH}$. By this, we mean that if the (mixed) state describing the system at time $0$ is $\rho,$ at time $t$ the state will have evolved to $U_H(t)\rho U_H^\dagger (t)$. Thus, a natural access model for Hamiltonians, introduced in \cite{huang2023heisenberg}, is to perform \emph{experiments} of the following kind: prepare a state $\rho,$ apply $U_H(t_1)$---that is, make a query to $U_H(t_1)$, which in a lab can be implemented by letting the system evolve for time $t_1$---, apply a unitary operator $V_1$ independent of $H$, query $U_H(t_2)$, apply a unitary operator $V_2$ independent of $H$, query $U_H(t_3)$, $\ldots$, and finally measure. There are several figures of merit to be optimized when performing a computational task in this access model. The one commonly considered the most important is the \emph{total evolution time}, which is the sum of all the times $t_i$ at which the algorithm queries $U_H(t_i).$ 

    As in prior work on Hamiltonian property testing from the dynamics \cite{she2022unitary,bluhm2024hamiltonian,arunachalam2025testing,kallaugher2025hamiltonian,gao2025quantum}, we will assume that the Hamiltonians have bounded operator norm,\footnote{We want to stress that our protocol assumes that the operator norm of the Hamiltonian is upper bounded, but that the total evolution time of our protocol is independent of such a bound. To be precise, our algorithm only needs to know that upper bound to be precise enough in a Trotterization step.}
    and we will consider the normalized Frobenius norm, given by $$\norm{H-H'}_{F}=\sqrt{\Tr[(H-H')^2]/2^n},$$ as the distance in the space of Hamiltonians. 
    The normalized Frobenius norm induces an average-case distance: if the normalized Frobenius norm between two Hamiltonians is small, then the expected values of observables measured on the two states generated by applying the time evolution of each Hamiltonian to a Haar-random state will be close (see \cite[Appendix A]{bluhm2024hamiltonian} and \cite[Section 7.2]{ma2024learning}). Such problems, where the object to be learned or tested is bounded in a worst-case norm, but the precision of the approximation is taken to be an average case norm, has also extensively been studied in the classical case, as for instance in the case of PAC learning of low-degree functions $f:\{0,1\}^n\to [-1,1]$ \cite{linial1993constant,eskenazis2022learning}. 
    
    Now, we are ready to state our first result (see \cref{thm:klocalcertificationdiagonal} for a formal and more detailed statement) on certifying $k$-local Hamiltonians from time evolution access.

    \begin{result}\label{result:CertificationTimeEvolution}
        Let $k$ be a constant. Let $H$ and $H_0$ be two $n$-qubit $k$-local Hamiltonians, where $H_0$ is known in advance, whereas $H$ is unknown. Then, there is an algorithm with access to a description of $H_0$ and to the time evolution of $H$ that only uses $ O(1/\eps)$ total evolution time, and, with high success probability, determines whether $\norm{H-H_0}_{F} \le \frac{\eps}{8\cdot 3^k}$ or $\norm{H-H_0}_{F}\geq \eps.$
    \end{result}

    \cref{result:CertificationTimeEvolution} is optimal up to constant factors, because $\Omega(1/\eps)$ evolution time is required to distinguish $H=\eps X$ from $H=-\eps X$ (see, for instance, \cite{kallaugher2025hamiltonian}). Several previous works considered testing Hamiltonian properties from time evolution access \cite{bluhm2024hamiltonian,arunachalam2025testing,kallaugher2025hamiltonian,gao2025quantum}, but the closest-to-optimal result for any Hamiltonian property testing task up to now was the $O(1/\eps^2)$ time evolution upper bound for testing locality given in \cite{kallaugher2025hamiltonian}, which is quadratically worse than the best lower bound $\Omega(1/\eps)$. 
    Thus, to our knowledge, Result~\ref{result:CertificationTimeEvolution} is the first optimal algorithm for testing a property of quantum Hamiltonians. Furthermore, we substantially improve upon the $\widetilde O(n^{3k/2}/\eps)$-evolution-time algorithm for certifying $k$-local Hamiltonians that can be obtained as a special case of the $\widetilde O(s^{3/2}/\eps)$-evolution-time algorithm for certifying Hamiltonians supported on at most $s$ Pauli operators given in \cite[Theorem 5.5]{gao2025quantum}. In a concurrent work, a $O(1/\eps^2)$ evolution-time algorithm is given for the case where $H$ is an arbitrary Hamiltonian and $H_0=0$ \cite{sinha2025improvedhamiltonianlearningsparsity}. Compared to this, our result constitutes a quadratic improvement when $H$ is promised to be a local Hamiltonian.

    \paragraph{The intolerant algorithm.} For simplicity, we will give a technical overview of the intolerant certification algorithm in this section, determining whether $H=H_0$ (EQUAL) or $\|H-H_0\|_F \ge \eps$ (FAR). We provide a detailed explanation of the tolerant algorithm in~\cref{sec:tolerant-cert}. Our algorithm relies on a novel application of the Bonami Lemma, also known as Hypercontractivity Theorem, for operators~\cite{bonami1970etude,montanaro2008quantum}. To later apply the Bonami Lemma, let us express $\Delta H$ as
    \begin{equation*}
        \Delta H = \sum_{s\in\{0,1\}^n} \lambda_s \ket{\psi_s}\bra{\psi_s},
    \end{equation*}
    where $\lambda_s$ are the eigenvalues and are the $\ket{\psi_s}$ are eigenvectors of $\Delta H$.
    
    We can simulate query access to the time evolution operator of $\Delta H$ thanks to Trotterization, which allows us to approximate $e^{-it\Delta H}$ up to arbitrarily small error by making queries to $e^{-itH}$ and $e^{-itH_0}$ \cite{childs2021theory}. In particular, we may perform Bell sampling on $e^{-i\Delta Ht}$. That is, we may sample from the distribution formed by the squares of the Pauli coefficients of $e^{-i\Delta Ht}$ by applying $e^{-i\Delta Ht}\otimes I^{\otimes n}$ on $n$ EPR pairs and measuring in the Bell basis. Then the probability of measuring identity is given by \cite[Theorem 1]{sinha2025improvedhamiltonianlearningsparsity}\footnote{For simplicity, we will use Bell sampling in the introduction, but, as shown in Lemma~\ref{lem:memorylessPaulisampling}, we can also estimate $I(t)$ without quantum memory, and thus we can also perform our certification algorithm without quantum memory.}
    \begin{equation*}
        I(t) = \frac{1}{4^n}\sum_{r,s\in \{0,1\}^n}\cos((\lambda_r-\lambda_s)t).
    \end{equation*} Note that $I(t)$ only depends on eigenvalue differences of $\Delta H$. In particular, to analyze $I(t)$
    , we consider the proportion of $\varepsilon$-separated eigenvalue pairs:
    \begin{equation*}
        \Lambda(\Delta H, \varepsilon) = \frac{1}{4^n} \sum_{r,s:|\lambda_r-\lambda_s| \ge \varepsilon} 1.
    \end{equation*}
    We first show that for a uniformly random $t\in [0,2/\eps]$,
    \begin{equation}\label{eq:Itupperbound}
        I(t)\leq 1-\frac{\Lambda}{4}
    \end{equation}
    with probability $\geq 1/3$ (see \cref{lem:spectral-condition}).
    
    Next, we show how to lower bound $\Lambda$ (and thus how to upper bound $I(t)$ thanks to \cref{eq:Itupperbound}) for $k$-local Hamiltonians. We can express the second and fourth moment of $F(r,s)=\lambda_r-\lambda_s$ in terms of $2$-norm and $4$-norm of $\Delta H$ as follows:
    \begin{align*}
        \frac{1}{4^n}\sum_{s,t}[|F(s,t)|^2] &= 2\|\Delta H\|_2^2 = 2\|\Delta H\|_F^2, \\
        \frac{1}{4^n}\sum_{s,t}[|F(s,t)|^4] &= 2\|\Delta H\|_4^4+6\|\Delta H\|_2^4, 
    \end{align*}
    where we use the convention $\norm{\Delta H}_p=(\tfrac{1}{2^n}\sum_{s}|\lambda_s|^p)^{1/p}$ (with this notation, $\norm{\Delta H}_2=\norm{\Delta H}_F$). Thanks to hypercontractivity, we can bound the $4$-norm of $\Delta H$ via the $2$-norm of $\Delta H$ as
    \begin{equation*}
        \|\Delta H\|_4^4 \le 9^k\|\Delta H\|_2^4.
    \end{equation*}
    From here, the Paley-Zygmund inequality yields
    \begin{equation*}
        \Lambda(\Delta H, \|\Delta H\|_F) = \Pr[|F(s,t)|\ge \|\Delta H\|_F] \ge \frac{1}{3}\cdot 9^{-k}.
    \end{equation*}
    Hence, we have that (see \cref{lem:boundonGamma} for a precise statement)
    \begin{equation}
        \norm{\Delta H}_F\geq \eps\implies \Lambda (\Delta H, \eps)=\Omega (9^{-k}),\label{eq:boundonGamma}
    \end{equation}
    which is the key equation that we need to show the correctness of our algorithm. 
    
    Now, we are ready to describe our simple algorithm. It consists of repeating the following $O(1)$ times:  sample a uniformly random $t\in [0,2/\eps]$, perform Bell sampling of $e^{-it\Delta H}$ a total of $O(9^{2k})$ times, and output FAR if \emph{many} non-identity Paulis are measured. If one of the repetitions outputs FAR, the final output is FAR; and if none of the repetitions outputs FAR, the final output is EQUAL. Note that the total evolution time is at most $O(9^{2k})\cdot (2/\eps)=O_{k}(1/\eps)$, as claimed.

    Next, we analyze the correctness of our algorithm. In case $\Delta H=0,$ the time evolution unitary $e^{-it\Delta H}$ equals the identity, so even with the Trotter error, we would only measure \emph{a few} non-identity Paulis. Thus, we will output EQUAL, as desired. If $\norm{\Delta H}_F\geq \eps$, then by \cref{eq:boundonGamma} we have that $\Lambda(\Delta H,\eps)=\Omega(9^{-k})$. From here and \cref{eq:Itupperbound}, it follows that $I(t)\leq 1-\Omega(9^{-k})$ with high probability for a random $t\sim [0,2/\eps]$. Thus, with high probability at least one of the repetitions of the iteration will output FAR, as desired.

    \paragraph{Other complexity measures and further remarks about our algorithms.} Our algorithm for Hamiltonian certification from dynamics is also efficient in other complexity measures such as query complexity. Notably, while the query complexity depends explicitly on an operator norm bound $C_{\text{op}}$ for the Hamiltonians $H$ and $H_0$, the overall evolution time and the number of quantum experiments of our algorithm are independent of $C_{\text{op}}$.
    Furthermore, our algorithm is robust against a constant amount of state-preparation and measurement (SPAM) errors. Finally, we do not need to perform Bell sampling, but just to estimate the probability of outputting identity when performing Bell sampling, which can be achieved without quantum memory~\cite{arunachalam2025testing} (see \cref{lem:memorylessPaulisampling} below). This is essential for near-term applications. We refer to~\cref{thm:klocalcertificationdiagonal} for the precise statement of our results, which includes all these details.

    \subsubsection{Learning thermal states}
    There is plethora of results about learning quantum Hamiltonians from access to the associated Gibbs states \cite{anshu2021sample, haah2022optimal, rouze2023learning, onorati2023efficient, bakshi2023learning, Gu2022Practical,chen2025learning}, which, as noted in \cite[Remark 18]{anshu2021sample}, implies learning the Gibbs state itself. 
    However, this Hamiltonian learning-based approach to to the problem of Gibbs state learning inherits a $\Omega(e^{\beta})$ lower bound on the sample complexity, i.e.\ the number of copies of the state used by the protocol~\cite[Theorem 1.2]{haah2022optimal}. Here, we circumvent this caveat and obtain a learning algorithm for Gibbs states that is sample-efficient with respect to every parameter (see \cref{theo:GibbsStateLearning} for a formal statement). 
    \begin{result}\label{result:learningGibbsStates}
        Let $\rho_H(\beta)$ be the Gibbs state of an unknown $n$-qubit local Hamiltonian $H$ at temperature $\beta$ with $|h_P|\leq 1$ for every $P$. Then, there is an algorithm that, with high probability, $\eps$-learns $\rho_H(\beta)$ in trace norm using only $\widetilde O(\poly(n)\beta^2/\eps^4)$ single copies of the state. 
    \end{result}
    In the case of $\beta=\poly(n)$, Result~\ref{result:learningGibbsStates} achieves Gibbs state tomography with exponential speedup over general state tomography, which requires $\Theta(4^n)$ copies of the state \cite{o2016efficient,haah2017sample}. Notably, our result, in contrast to all the aforementioned prior works, only requires $k$-locality of the Hamiltonian, and no further assumptions (such as every qubit being acted on by a constant number of Pauli operators) are made. Alas, our algorithm achieving  Result~\ref{result:learningGibbsStates} is not time-efficient, similarly to the first algorithm for learning quantum Hamiltonians from Gibbs states \cite{anshu2021sample}.

    The time-inefficiency is intrinsic to the $\eps$-covering net argument underlying the proof of Result~\ref{result:learningGibbsStates}. The proof starts by establishing the following inequality, which is a consequence of Pinsker's inequality (see Lemma~\ref{lem:dSKLGibbsStates} for a proof):
    \begin{equation}\label{eq:dsklintro}
        \norm{\rho_H(\beta)-\rho_{H'}(\beta)}_{\tr}
        \leq \sqrt{2 \beta \Tr[(\rho(\beta)-\rho'(\beta))(H'-H)]}
        =O(\beta n^k\max_{P:  |P|\le k} |h_P-h_P'|)
    \end{equation}  
    for every pair of $k$-local Hamiltonians $H, H'$. This bound ensures that the set $$\mathcal S_\eta=\{\rho_H(\beta):\, H\in\mathcal H_\eta\}$$ of Gibbs states, where
    $$\mathcal H_\eta=\{H:H\ k\text{-local Hamiltonian with }h_P\in \eta\mathbb Z\cap [-1,1]\ \forall P\},$$
    is an $\eps$-covering net of the set of $k$-local Gibbs states when taking $\eta$ of the order $\eps/(\beta n^k)$. Next, we note that the observables $H-H'$ for $H,H'\in\mathcal H_\eta$ are sums of $k$-local Pauli strings. Hence, using classical shadows \cite{huang2020predicting} (see \cref{theo:Shadows}), we can simultaneously obtain accurate estimates $\Delta_{H,H'}$ for all $\Tr[\rho(H-H')]$ in a sample-efficient manner, where $\rho$ is the state to be learned. If these estimates were exact and the state belonged to the net, then by \cref{eq:dsklintro} one would be able to identify the state. The rest of the proof consists of showing that, even if the state does not belong to the net and with an error in the estimates, the state
    \begin{equation*}
        \rho'=\text{argmin}_{\rho'\in\mathcal S_\eta}\max_{H,H'\in\mathcal H_\eta}|\Delta_{H,H'}-\Tr[\rho'(H-H')]|
    \end{equation*}
    satisfies $\norm{\rho-\rho'}_{\tr}\leq \eps$ with high probability \cite{Yatracos1985Jun,buadescu2021improved}.
    
    \subsubsection{Certifying thermal states}
    We also show that quantum state certification of local Gibbs states can be made sample and time-efficient with respect to all parameters, resolving a question by Anshu \cite[Section 2]{anshu2022some} (see \cref{theo:certifyingGibbsStates} for a formal statement).

    \begin{result}\label{result:certifyingGibbsStates}
        Let $\rho_H(\beta)$ and $\rho_{H_0}(\beta)$ be the Gibbs states of an $n$-qubit local Hamiltonian $H$ and $H_0$ at temperature $\beta$ with $|h_P|,|(h_0)_P|\leq 1$ for every $P$. Then, there is an algorithm that, with high probability, decides whether $\rho_{H}(\beta)=\rho_{H_0}(\beta)$ or $\norm{\rho_{H}(\beta)-\rho_{H_0}(\beta)}_{\tr}\geq \eps$ using only $\widetilde O(\poly(n)\beta^2/\eps^4)$ single copies of the states. 
    \end{result}
   In the case of $\beta=\poly(n)$, Result~\ref{result:certifyingGibbsStates} shows an exponential speedup for Gibbs state certification over general state certification, which requires $\Theta(2^n)$ copies of the state \cite{odonnell2015quantum, buadescu2019quantum}. Furthermore, the algorithm behind Result~\ref{result:certifyingGibbsStates} is time-efficient in every parameter (in contrast with Result~\ref{result:learningGibbsStates}). The proof of Result~\ref{result:certifyingGibbsStates} is based on an inequality of the kind of \cref{eq:dsklintro}, which we expect to be useful in other scenarios, and which has seen applications in the classical literature~\cite{santhanam2012information,daskalakis2019testing}. 

    \subsection{Related work}

    \textbf{Hamiltonian certification.} As we mentioned, the only prior work directly focused on quantum Hamiltonian certification was the one of Gao, Ji, Wang, Yu, and Zhao \cite{gao2025quantum}. Their main result shows that one can certify Hamiltonians that are $s$-sparse in the Pauli basis with $O(s^{3/2}/\eps)$ evolution time. They also show that if one allows access to the inverse of the time evolution (which lacks a clear physical motivation, as it amounts to allowing time to flow backwards), then one can achieve an optimal $O(1/\eps)$ for arbitrary Hamiltonians. Our results show that neither reverse- nor controlled-time evolution access are needed to achieve Heisenberg scaling in Hamiltonian certification under a locality promise.

    \paragraph{Hamiltonian learning.} Prior approaches to $O(1)$-local Hamiltonian learning could be applied for Hamiltonian certification, but they would yield suboptimal results. Indeed, the approaches focused on attaining Heisenberg limited scaling require a sparsity assumption, and as general local Hamiltonians are $\poly(n)$-sparse (in fact, there can be $\poly(n)$ Pauli terms involving a given particle), these works yield evolution times of the order $O(\poly(n)/\eps)$ \cite{huang2023heisenberg,bakshi2024structure}. The Hamiltonian learner of \cite{arunachalam2025testing} does not assume sparsity and has no dependence on $n$, but the total evolution time of their protocol is $O(1/\poly(\eps))$, which is not optimal for certification either.
    
    \paragraph{Hamiltonian property testing.} Although the literature on Hamiltonian learning is much larger than the one on Hamiltonian property testing, several recent works studied questions about Hamiltonian property testing from dynamics \cite{aharonov2022quantum,laborde2022quantum,she2022unitary,bluhm2024hamiltonian,arunachalam2025testing,kallaugher2025hamiltonian,gao2025quantum,sinha2025improvedhamiltonianlearningsparsity}. They focus on testing for structural properties of Hamiltonians, such as locality or sparsity, but none of them achieves a provably optimal result. Thus, Result~\ref{result:CertificationTimeEvolution} is the first optimal result in Hamiltonian property testing.

    \paragraph{Certification in other settings.} The problem of certification has been extensively studied in other contexts, both classical and quantum. Maybe the most fundamental variant of certification is the uniformity testing of distributions, which is the problem of testing from samples whether a probability distribution is uniform or far from it (see \cite{paninski2008coincidence} for the first optimal algorithm for this problem, and the survey of Canonne for a discussion on the plethora of approaches to the problem \cite{canonne2022topics}). The quantum access version of that problem has also been studied (and remains open), where one has access to states whose amplitudes square to the probability distribution \cite{bravyi2011quantum,chakraborty2010newresults,luo2024succinct,canonne2024uniformity}. The fully quantum variant of this problem is the problem of testing whether a state is the maximally mixed one, and its resolution was central in the tomography developments of O'Donnell and Wright \cite{odonnell2015quantum,o2016efficient}. Finally, there are also works on the certification of other quantum objects, such as unitaries and channels \cite{fawzi2023quantum-channel-certification,Rosenthal2024Sep,jeon2025query}.

    \paragraph{Improved Bell sampling bounds.}
    Sinha and Tong~\cite{sinha2025improvedhamiltonianlearningsparsity} established sharp short-time bounds on Bell sampling. They show that for any traceless Hamiltonian $H$ with operator norm $\|H\|_{\mathrm{op}} \le L$, the probability $I(t)$ of obtaining the Bell outcome $I^{\otimes n}$ after time evolution $e^{-iHt}$ satisfies
        \begin{equation*}
            1 - t^2 \|H\|_F^{\,2}  \le I(t) \le 1 - 2ct^2\|H\|_F^{\,2},
        \end{equation*}
        for any $c \in (0,1/2)$, provided that $t \le {t^*(c)}/{(2L)}$, where $t^*(c) \in (0,2\pi)$ is defined implicitly by $\cos(t^*(c)) = 1 - c(t^*(c))^2$. These inequalities capture the quadratic short-time decay of $I(t)$, but the admissible time window  
        $t \le t^*(c)/(2L)$ is governed by the operator norm of $H$, and therefore remains in the regime  
        \(t = O(1/\|H\|_{\mathrm{op}})\).  
        This is too restrictive for certification tasks requiring evolution times on the order of \(1/\varepsilon\). 
        
        A natural question is whether similar two-sided control of $I(t)$ can be extended to the longer-time regime $t = \Theta(1/\varepsilon)$, under the promise $\|H\|_F \ge \varepsilon$. If such bounds were available, then one could directly distinguish the cases  
        \(\|H\|_F = 0\) and \(\|H\|_F \ge \varepsilon\) using only $\Theta(1/\varepsilon)$ total evolution time. We show that this long-time extension is indeed achievable for $O(1)$-local Hamiltonians. While the Sinha–Tong~\cite{sinha2025improvedhamiltonianlearningsparsity} bounds are inherently restricted to the short-time region determined by the operator norm, we derive new estimates that remain valid up to times of order \(1/\varepsilon\). These improved Bell sampling bounds establish that constant-locality Hamiltonians admit optimal $\Theta(1/\varepsilon)$-time intolerant certification.
    
    \subsection{Discussion and open problems}

    In this work, motivated by the importance of local Hamiltonians to various areas of quantum science, we have explored the tasks of certifying and learning such Hamiltonians. 
    First, we have given an algorithm for local Hamiltonian certification with optimal total evolution time, thus providing to our knowledge the first optimal bound for any Hamiltonian property testing task in the time-evolution access model.
    Next, we have shifted our focus from the Hamiltonians themselves to the associated Gibbs states. For both learning and certification, this change of perspective allowed us to develop fully sample-efficient---and, in the case of certification, even time-efficient---algorithms.
    This in particular overcomes a known exponential-in-$\beta$ lower bound on learning Hamiltonians from access to copies of the Gibbs state, thus (re-)positioning Gibbs state learning and testing as tasks of independent interest alongside Hamiltonian learning and testing.

    We conclude this introduction by posing several open questions arising from our results:
    \vspace{-0.5cm}
    \paragraph{Optimal certification from dynamics without locality assumptions.} Our main result exploits the locality
        of the underlying Hamiltonians, and in the $O(1)$-local regime we show that optimal $O(1/\eps)$ total
        evolution time is achievable even without inverse or controlled evolution. A fundamental open problem is whether such $O(1/\eps)$ evolution-time certification remains possible without any locality assumptions on
        $H$ or $H_0$. Previous optimal certification results for general (non-local) Hamiltonians, such as those
        of \cite{gao2025quantum}, rely on access to inverse- or controlled-time evolution. Resolving this question would clarify whether locality is merely a technical convenience
        in our analysis or whether it plays a fundamental role in enabling evolution-time-optimal protocols.

    \paragraph{Time-efficient Gibbs state learning.} The seminal result of learning Hamiltonians via access to the Gibbs state of \cite{anshu2021sample} was only sample-efficient (with respect to $n$), and it was made time-efficient in a series of follow-up works \cite{haah2022optimal,bakshi2023learning}. Similarly, our Result~\ref{result:learningGibbsStates} is, to our knowledge, the first algorithm for learning Gibbs states that is sample-efficient in all parameters. It is thus natural to wonder: Is there an algorithm for learning Gibbs states of $k$-local Hamiltonians that is both sample- and time-efficient in every parameter? 
    
    \paragraph{Optimal Gibbs state certification.} Result~\ref{result:certifyingGibbsStates} is already efficient in both sample and time complexity, but we lack a matching lower bound. Even for its classical counter-part \cite{daskalakis2019testing}, the precise complexity of local Gibbs states seems to be unknown.
        Thus, we ask: What is the optimal sample-complexity of certifying local Gibbs states?

    \paragraph{Certification from local probes.} Recent advances in Hamiltonian learning suggest that even
    severely restricted access, such as single-site or constant-size probes of a system undergoing time
    evolution, can still yield meaningful reconstruction guarantees \cite{chen2025quantumprobe}. This motivates a natural
    open problem for certification: can you efficiently certify a $k$-local Hamiltonian
    using only the dynamics observed through a strictly local probe? Another central question is whether
    there exists a quantitative tradeoff between probe size and certification complexity. For example, a
    constant-size probe may or may not suffice to achieve optimal $\Theta(1/\eps)$ evolution time, and a larger
    probe may be fundamentally required to detect the spectral features that distinguish $H$ from $H_0$.
    Understanding this relationship remains open and would help bridge the gap between theoretically
    optimal certification protocols and experimentally realistic settings with limited local access.

    \paragraph{Optimal sparsity testing.}     Our algorithm for certification from dynamics relies on spectral gap arguments for Hamiltonians. These ingredients resemble those used in recent work on learning and testing $s$-sparse Hamiltonians by Sinha and Tong \cite{sinha2025improvedhamiltonianlearningsparsity}. It
    is therefore natural to ask whether similar techniques can lead to optimal algorithms for Hamiltonian sparsity testing, maybe under the promise of locality.

    \paragraph{Beyond Hamiltonians.} Our analysis focuses on closed-system dynamics generated by time-independent Hamiltonians. Many realistic devices, however, are better modeled by quantum channels or Lindbladian evolutions, such as Pauli channels induced by noise processes. A natural direction is to investigate whether the techniques developed here 
    can be adapted to certify dynamical maps such as Pauli channels or Linbladians.

    \paragraph{Note.} This manuscript subsumes the optimal certification protocol for constant-local quantum Hamiltonians due to Lee and Shin~\cite{lee2025optimal}, as well as the certifying and learning algorithms for quantum Ising Hamiltonians developed by Bluhm, Caro, Escudero-Gutiérrez, Oufkir, and Rouzé~\cite{bluhm2025certifying}.

\section{Preliminaries}
We start by introducing some notation. $I,X,Y$ and $Z$ are the 1-qubit Pauli matrices, and a tensor product of these matrices is called a Pauli string.  Any matrix $A$ acting on $n$ qubits is a matrix of $(\C^{2\times 2})^{\otimes n}$. Such a matrix can be expressed as a linear combination of Pauli strings via its Pauli expansion $A=\sum_{P\in\{I,X,Y,Z\}^{\otimes n}}a_PP.$ Here, $a_P$ are the Pauli coefficients and they are determined by $$a_P=\frac{1}{2^n} \Tr[PA].$$ 
A Pauli string is called $k$-local if it acts as identity in all but at most $k$ qubits. The number of $k$-local Pauli strings is at most 
\begin{equation}\label{eq:numberofklocalPaulis}
    100n^k,
\end{equation}
because
\begin{align*}
        \sum_{l=0}^k 3^l\binom{n}{l}\leq \left\{\begin{array}{ll}
             (k+1)3^k\left(en/k\right)^k\leq 100 n^k             & \text{if } k<n/2 \\ 
            4^n\leq 20 n^{n/2}\leq 20n^k & \text{if }k\geq n/2
        \end{array}\right.,
\end{align*}
where we have used that $(3e/k)^k(k+1)<100$ and $4^n\leq 20 n^{n/2}$ for every $n,k\in\mathbb N$. Given a matrix $A$ acting on $n$ qubits, $\norm{A}_{\op}$ denotes the usual operator norm, i.e., the largest singular value of $A;$ $\norm{A}_{\tr}$ is the trace norm, i.e., the sum of the singular values of $A$; and $\norm{A}_{F}=\sqrt{\Tr[A^\dagger A]/2^n}$ is the normalized Frobenius norm. The Pauli strings are an orthonormal basis with respect to the inner product $\langle A,B\rangle=\Tr[A^{\dagger}B]/2^n.$ In particular, Parseval's identity states that
\begin{equation*}
    \norm{A}_{F}=\sqrt{\sum_{P\in\{I,X,Y,Z\}^{\otimes n}}|a_P|^2}.
\end{equation*}
A more general version of Parseval's identity is Plancherel's identity, which states that 
\begin{equation*}
  \langle A,B\rangle\equiv   \frac{\Tr[A^\dagger B]}{2^n}=\sum_{P\in\{I,X,Y,Z\}^{\otimes n}}\bar{a}_Pb_P\, ,
\end{equation*}
where for $z\in\C$, $\bar z$ denotes the complex conjugate of $z.$ We use $\widetilde \Omega(\cdot)$ and $\widetilde O(\cdot)$ to hide polylogarithmic factors of the quantities inside the parentheses.

\subsection{Hamiltonians}
An $n$-qubit Hamiltonian is a self-adjoint matrix acting on $n$ qubits. In particular, a matrix $A$ is a Hamiltonian if and only if $a_P\in \mathbb R$ for every $P\in\{I,X,Y,Z\}^{\otimes n}$. A Hamiltonian $H$ is $k$-local if $h_P=0$ for every $P=P_1\otimes \dots \otimes P_n$ such that $|P|:=|\{i\in [n]:\ P_i\neq I\}|>k.$ We will assume that Hamiltonians are traceless, meaning that $h_{I^{\otimes n}}=\Tr[H]/2^n$=0. This is without loss of generality, because two Hamiltonians that only differ in a multiple of identity determine the same time evolution operators and the same Gibbs states. 

\subsubsection{Access via dynamics} 
Hamiltonians govern the dynamics of (closed) quantum systems according to the Schr\"odinger equation. In particular, if a quantum system governed by a time-independent Hamiltonian $H$ and the state describing the system at time $0$ is $\rho,$ 
at time $t$ the state will have evolved to $U_H(t)\rho U_H^\dagger (t)$, where $U_H(t)=\exp(-itH)$ is the time evolution operator of $H$ at time $t$. 

Thus, a natural access model for Hamiltonians is to perform \emph{experiments} of the following kind: prepare a state $\rho,$ apply $U_H(t_1)$---that is, make a query to $U_H(t_1)$, which in a lab can be implemented by letting the system evolve for time $t_1$---, apply a unitary operator $V_1$ independent of $H$, query $U_H(t_2)$, apply a unitary operator $V_2$ independent of $H$, query $U_H(t_2)$,$\dots$ and finally measure. 
In this access model, there are different potentially relevant figures of merit. The one usually considered as the most important is the \emph{total evolution time}, which is the sum of all times $t_i$ at which the algorithm queries $U_H(t)$. Other figures of merit that we will also keep track of are the \emph{number of experiments}, the \emph{number of queries,} the \emph{time resolution} (i.e., the minimum time at which the algorithm queries the time evolution operator), the \emph{classical post-processing time}, and the number of \emph{ancilla qubits}.

Finally, our algorithms will also be robust to \emph{state-preparation and measurement (SPAM) error}. Following \cite[Definition 4]{ma2024learning}, an experiment suffers from an $\eps$-amount of SPAM error if the error channels applied after the initial state preparation and before the first query and the error channels after the last query and before the measurement induce in total $\eps$ error in diamond norm. We will say that an algorithm is \emph{robust} to an $\eps$ amount of SPAM error (or any other error) if the performance guarantees of the algorithm do not change in the presence of that error, maybe after increasing the complexities by constant factors. 

\subsection{Access via Gibbs state}
Hamiltonians also determine the equilibrium states of quantum systems. In particular, if a quantum system is governed by a Hamiltonian $H$, then the equilibrium state of the system at inverse temperature $\beta>0$ is  the \emph{Gibbs state} given by $\rho(\beta)=e^{-\beta H}/\Tr[e^{-\beta H}]$. 

An alternative access model for Hamiltonians is hence to perform measurements on copies of the Gibbs state of the Hamiltonian. The main figure of merit in this model is the \emph{sample complexity,} i.e., the number of copies of the Gibbs state that the algorithm accesses. Other important figures of merit that we will keep track of are the \emph{maximum number of copies that the algorithm measures coherently} and the \emph{classical post-processing time}. In particular, we say that an algorithm uses \emph{single copies} of the state if it measures one copy of the state at a time.

All of our results in this access model use the following upper bounds on the trace distance between Gibbs states, which are well-known in the classical literature \cite{santhanam2012information,daskalakis2019testing}, and similar bounds have been used in the quantum literature \cite{anshu2021sample,fanizza2024efficient}. 

\begin{lemma}\label{lem:dSKLGibbsStates}
    Let $\rho(\beta)$ and $\rho'(\beta)$ be Gibbs states of two $k$-local Hamiltonians $H$ and $H'$ acting on $n$ qubits. Then,
    \begin{equation}\label{eq:trnormforGibbs0}
        \norm{\rho(\beta)-\rho'(\beta)}_{\tr}\leq \sqrt{2 \beta \Tr[(\rho(\beta)-\rho'(\beta))(H'-H)]}.
    \end{equation}
    In particular,
    \begin{equation}\label{eq:trnormforGibbs1}
        \norm{\rho(\beta)-\rho'(\beta)}_{\tr}\le  200\beta n^k \sup_{|P|\leq k}|h_P-h_P'|.
    \end{equation}
    Furthermore, if $|h_P|,|h'_P|\leq 1$ for every $P\in\{I,X,Y,Z\}^{\otimes n}$, then 
    \begin{equation}\label{eq:trnormforGibbs2}
        \norm{\rho(\beta)-\rho'(\beta)}_{\tr}\le \sqrt{400\beta n^k \sup_{|P|\leq k}2^n|\rho(\beta)_P-\rho'(\beta)_P|}.
    \end{equation}
\end{lemma}
\begin{proof}
    We start using Pinkser inequality to upper bound the trace norm as 
    \begin{align*}
        \norm{\rho(\beta)-\rho'(\beta)}_{\tr}\le\sqrt{2\Tr[\rho(\beta)(\log\rho(\beta)-\log\rho'(\beta))]+2\Tr[\rho'(\beta)(\log\rho'(\beta)-\log\rho(\beta))]}.
    \end{align*}
    Now, expanding the right-hand side and using that $\log \rho(\beta)=-\beta H-Z(\beta)$, where $Z(\beta)=\Tr[e^{-\beta H}]$, we arrive at 
    \begin{align}
        \norm{\rho(\beta)-\rho'(\beta)}_{\tr}&\le \sqrt{2\beta\Tr[(\rho(\beta)-\rho'(\beta))(H'-H)]}.\label{eq:trnormforGibbs3}
    \end{align}
    This proves \cref{eq:trnormforGibbs0}.

    Now, we focus on proving \cref{eq:trnormforGibbs1}. On the one hand, using \cref{eq:trnormforGibbs3} and that $|\Tr[A^\dagger B]|\leq \norm{A}_{\tr} \norm{B}_{\op}$ we get 
    \begin{align*}
        \norm{\rho(\beta)-\rho'(\beta)}_{\tr}
        &\le \sqrt{2\beta\norm{\rho(\beta)-\rho'(\beta)}_{\tr}\norm{H-H'}_{\op}}\, ,
    \end{align*}
    so 
    \begin{align}
        \norm{\rho(\beta)-\rho'(\beta)}_{\tr}
        &\le 2\beta\norm{H-H'}_{\op},\label{eq:1}
    \end{align}
    On the other hand, by triangle inequality and \cref{eq:numberofklocalPaulis}, we have that 
    \begin{align*}
        \norm{H-H'}_{\op}&\leq  100 n^k \sup_{|P|\leq k}|h_P-h_P'|,
    \end{align*}
    which combined with \cref{eq:1} proves \cref{eq:trnormforGibbs1}. 

    Now, we focus on proving \cref{eq:trnormforGibbs2}. Using \cref{eq:trnormforGibbs3} and Plancherel's identity we arrive at 
    \begin{align*}
        \norm{\rho(\beta)-\rho'(\beta)}_{\tr}&\le\sqrt{2\beta\sum_{|P|\leq k}2^n(\rho(\beta)_P-\rho'(\beta)_P)(h_P-h'_P)}.
    \end{align*}
    Hence, as $|h_P|,|h_P'|\leq 1$ and by \cref{eq:numberofklocalPaulis}, we have that 
    \begin{align*}
        \norm{\rho(\beta)-\rho'(\beta)}_{\tr}&=\sqrt{400\beta n^k\sup_{|P|\leq k}2^n|\rho(\beta)_P-\rho'(\beta)_P|},
    \end{align*}
    which proves \cref{eq:trnormforGibbs2}.
\end{proof}

\subsubsection{Trotterization}
Given access to $e^{-itA}$ and $e^{-itB}$ for two Hamiltonians $A$ and $B$ and arbitrary times $t$, Trotterization allows us to implement $e^{-it(A+B)}$ up to arbitrary error while also preserving the total time evolution and without using extra qubits. Thus, to analyze the number of experiments and the total time evolution required by our algorithms, if we have access to $e^{-itA}$ and $e^{-itB}$, we may assume access to $e^{-it(A+B)}$. However, the number of queries and the time resolution change. To be more precise, we will use the following result. 

\begin{theorem}[{\cite[Corollary 2]{childs2021theory}}]\label{theo:trotterization}
    Let $t>0$, let $\eps>0$, let $H,H_0$ be Hamiltonians acting on $n$-qubits, and let $c=\max\{\norm{H}_{\op},\norm{H_0}_{\op}\}$. Let $l=\left\lceil O\left(\sqrt{(ct)^3/\eps_{\operatorname{Trott}}}\right)\right\rceil$ and define $V=(e^{-itH/2l}  e^{itH_0/l}e^{-itH/2l})^l$. Then, 
        \begin{equation*}
        \norm{e^{-it(H-H_0)}-V}_{\op}
        \leq \eps_{\operatorname{Trott}}.
    \end{equation*}
\end{theorem}

\subsubsection{Useful subroutines}

We will use the following lemma that was proved in \cite[Lemma 3.3]{arunachalam2024testing}.\footnote{We note that in \cite[Lemma 3.3]{arunachalam2024testing} the authors only explicitly analyze the query complexity of their algorithm, but the analysis of the remaining figures of merit is straightforward.} Before stating it, we recall that a stabilizer subgroup of the group of Pauli matrices $\mathcal{S}\subseteq\{I,X,Y,Z\}^{\otimes n}$ is an abelian subgroup that does not contain $-I$. A stabilizer state corresponding to a stabilizer subgroup $\mathcal{S}$ of dimension $k\le n$ is defined as
\begin{align*}
\rho_{\mathcal{S}}:=\frac{1}{2^n}\sum_{P\in\mathcal{S}}P\,.
\end{align*}

\begin{lemma}\label{lem:memorylessPaulisampling}
    Let $U$ be an $n$-qubit unitary, and let $\eps,\delta>0$. There is a memory-less algorithm that makes $O\big(\log(1/\delta)/\eps^2\big)$ experiments that provides an estimate $|u'_{I^{\otimes n}}|^2$ such that 
    $$
    \big||u_{I^{\otimes n}}|^2-|u_{I^{\otimes n}}'|^2\big|\leq\eps
    $$
    with probability $\geq 1-\delta$. Furthermore, the algorithm makes only one query to $U$ per experiment, only stabilizer states, and only performs Clifford measurements. In addition, it is robust to $\eps/3$ amount of SPAM errors and $\eps/3$ error in diamond norm per query of $U$, and requires only $O\big(\log(1/\delta)/\eps^2\big)$ classical post-processing time.
\end{lemma}

We will also need to perform classical shadow tomography.

\begin{theorem}[Clifford shadows \cite{huang2020predicting}]\label{theo:Shadows}
    Let $\rho$ be an $n$-qubit state and let $k\in\mathbb N$, $\eps>0$ and $\delta>0.$ Then, performing random Pauli measurements on $$O
    \left(\frac{3^kk\log(n/\delta)}{\eps^2}\right)$$ single copies of $\rho$ suffices to obtain estimates $\widetilde \rho_P$ that with probability $\geq 1-\delta$ satisfying $$2^n|\rho_P-\widetilde \rho_P|\leq \eps$$ for every $|P|\leq k$. The classical post-procesing time is $O
    \left((3n)^kk\log(n/\delta)/\eps^2\right)$.
\end{theorem}
\subsection{Bonami Lemma}
We will use the quantum version of Bonami Lemma~\cite{bonami1970etude} proved by Montanaro and Osborne~\cite[Corollary 8.9]{montanaro2008quantum}. 
\begin{theorem}\label{lem:BonamiLemma}
    Given a $k$-local Hamiltonian $H$ on $n$ qubits and $p\geq 2,$ it holds that
    $$\|H\|_p \le (p-1)^{\frac{k}{2}}\|H\|_2,$$
    where $\|H\|_p$ is defined as 
    $$\|H\|_p=\left(\frac{\Tr[|H|^p]}{2^n}\right)^{1/p}.$$ 
\end{theorem}
Note that $\|H\|_2 = \|H\|_F$.

\subsection{Inequalities of probabilities}
Here we list the well-known probability bounds that we will use. 

\begin{lemma}[Paley--Zygmund inequality]\label{lem:p-z}
Let $Z\ge 0$ be a random variable with $\mathds{E}[Z^2]<\infty$ and let $0\le \theta\le 1$. Then
\begin{equation*}
    \Pr\left[ Z > \theta\,\mathds{E}[Z] \right]
    \ge
    (1-\theta)^2 \cdot \frac{\mathds{E}[Z]^2}{\mathds{E}[Z^2]}.
\end{equation*}
\end{lemma}

\section{Hamiltonian certification via access to time-evolution}\label{sec:tolerant-cert}
In this section, we propose an algorithm that uses access to time-evolution to certify whether an unknown local Hamiltonian $H$ is close to or far to a known local Hamiltonian $H_0$. We introduce a protocol that is optimal for any constant $k$. The correctness of the protocol relies on a bound on the probability of measuring identity when performing Bell sampling. Such a bound is possible for $k$-local Hamiltonians thanks to the Bonami Lemma.

\paragraph{Upper bounding the probability of measuring identity in Bell sampling.}\label{sec:sufficient-conditions}
We recall that $\Delta H=H-H_0$ can be expressed as
\begin{equation}\label{eq:DiagHam}
    \Delta H = \sum_{s\in\{0,1\}^n} \lambda_s\ketbra{\psi_s}{\psi_s}, 
\end{equation}
where $\lambda_s$ and $ \ket{\psi_s}$ are the eigenvalues and eigenvectors of $\Delta H$. The probability of measuring identity when performing Bell sampling is given by \cite[Theorem 1]{sinha2025improvedhamiltonianlearningsparsity}
\begin{equation}
    I(t) = \frac{1}{4^n}\sum_{r,s\in \{0,1\}^n}\cos((\lambda_r-\lambda_s)t).\label{eq:ProbIdBellSampling}
\end{equation} Note that $I(t)$ only depends on eigenvalue differences $F(r,s):=\lambda_r-\lambda_s$ of $\Delta H$. In particular, to analyze $I(t)$, we consider the proportion of $\varepsilon$-separated eigenvalue pairs:
\begin{equation*}
    \Lambda(\Delta H, \varepsilon) := \frac{1}{4^n} \sum_{r,s:|F(r,s)| \ge \varepsilon} 1.
\end{equation*}
First, we will show that a lower bound in $\Lambda$ implies an upper bound (smaller than 1) for $I(t)$ with high probability over some random $t.$ Second, we will prove that if the Hamiltonians are far away from each other, one can lower bound $\Lambda,$ and thus upper bound $I(t)$ away from 1. This will be enough for certification, because in the case that $H$ and $H_0$ are close, we will prove a lower bound of $I(t)$ for every $t.$

\begin{lemma}[Bounding $I(t)$ in terms of $\Lambda$]\label{lem:spectral-condition}
Fix $\varepsilon>0$. Let $\Delta H$ be an $n$-qubit Hamiltonian and $\Lambda := \Lambda(\Delta H,\varepsilon)$.
Then, if we pick $t\in [0,2/\eps]$ uniformly at random, with probability at least 1/3, we have that 
\begin{equation*}
    I(t) \le 1 - \frac{\Lambda}{4}
\end{equation*}
\end{lemma}

\begin{proof}
Let $t$ be uniformly random in $[0,2/\varepsilon]$. Let $\{\lambda_s\}$ be as in \cref{eq:DiagHam}, and recall that $F(r,s)=\lambda_r-\lambda_s$. Define $I_{\ge \varepsilon}(t):= \mathds{E}_{r,s:|F(r,s)|\ge \varepsilon}\left[ \cos(F(r,s) t) \right]$, where the expectation is with respect to the uniform distribution over pairs $(r,s)$ satisfying $|F(r,s)|\ge \varepsilon$.
Then,
\begin{equation}\label{eq:upboundonIt}
    I(t) \le (1-\Lambda) + \Lambda\, I_{\ge \varepsilon}(t).
\end{equation}
For any pair $(r,s)$ with $|\lambda_r-\lambda_s| \ge \varepsilon$, the average of $\cos(F(r,s) t)$ over $t \sim \mathrm{Unif}[0,2/\varepsilon]$ is
\begin{equation*}
    \mathds{E}_t\left[\cos(F(r,s) t)\right]
        = \frac{\varepsilon}{2}\cdot\frac{\sin(2F(r,s)/\varepsilon)}{F(r,s)},
\end{equation*}
which satisfies
\begin{equation*}
    \left|\mathds{E}_t[\cos(F(r,s) t)]\right|
        \le \frac{\varepsilon}{2|F(r,s)|}
        \le \frac{1}{2},
\end{equation*}
since $|F(r,s)| \ge \varepsilon$. In particular, $\mathbb E_t[I_{\ge \eps}(t)]\leq 1/2$, so $$\Pr_t[I_{\geq \eps}(t)\geq 3/4]\leq 2/3.$$
In other words, $\Pr_t[I_{\geq \eps}(t)\leq 3/4]\geq 1/3.$ By \cref{eq:upboundonIt} this implies that $\Pr_t[I(t)\leq 1-\Lambda/4]\geq 1/3.$
\end{proof}

Next, we show that a lower bound on $\norm{\Delta H}_F$ implies a lower bound on $\Lambda.$ Crucially, for this we use the locality assumption. Without this assumption, it may happen that all eigenvalues of $\Delta H$ are equal, and thus $\Lambda(\Delta H,\eps)=0$ for every $\eps>0,$ even when $\Delta H\neq 0$. However, thanks to the Bonami Lemma, we can ensure that the spectrum of $\Delta H$ for local Hamiltonians is well-spread. 

\begin{lemma}\label{lem:boundonGamma}
    Let $\Delta H$ be a $k$-local Hamiltonian. Then, $$\Lambda(\Delta H,\norm{\Delta H}_F)\ge\frac{1}{3\cdot 9^{k}}.$$
\end{lemma}

\begin{proof}
    We start by noting that we can recast $\Lambda(\Delta H,\norm{\Delta H}_F)$ as $$\Pr_{r,s}[|F(r,s)|\geq \norm{\Delta H}_F]=\Pr_{r,s}[|F(r,s)|^2\geq \norm{\Delta H}_F^2],$$ where the probability is over uniformly random $r,s$. Thus, we can lower bound $\Lambda$ as
    \begin{equation}
        \Lambda(\Delta H,\|\Delta H\|_F)=\Pr_{r,s}[|F(r,s)|^2\geq \frac{1}{2}\mathbb E_{r,s}[|F(r,s)|^2]]\geq \frac{1}{4}\frac{\mathbb E_{r,s}[|F(r,s)|^2]^2}{\mathbb E_{r,s}[|F(r,s)|^4]},\label{eq:applicationOfPaleyZygmund}
    \end{equation}
    where we have used the Paley--Zygmund inequality (see \cref{lem:p-z}) with $Z=F$ and $\theta=1/2$. With simple calculations we deduce that 
    $$\mathbb E_{r,s}[|F(r,s)|^2]=2\mathbb{E}_s[\lambda_s^2]=2\|\Delta H\|_2^2 = 2\|\Delta H\|_F^2,$$ and
    $$\mathbb E_{r,s}[|F(r,s)|^4]=2\mathbb{E}_s[\lambda_s^4]+6\mathbb{E}_s[\lambda_s^2]^2=2\|\Delta H\|_4^4+6\|\Delta H\|_2^4,$$
    because $\sum_s\lambda_s=\Tr[\Delta H]=0$. 
    
    Now, we use that $\Delta H$ is $k$-local, so we can apply the Bonami Lemma (see \cref{lem:BonamiLemma}) to $\Delta H$, yielding 
    \begin{equation*}
        \|\Delta H\|_4^4 \le 9^k\|\Delta H\|_2^4.
    \end{equation*}
    Hence, we have
    $$\mathbb E_{r,s}[|F(r,s)|^4]=2\|\Delta H\|_4^4+6\|\Delta H\|_2^4 \le 2(9^k+3)\|\Delta H\|_2^4 = \frac{9^k+3}{2}\mathbb E_{r,s}[|F(r,s)|^2]^2.$$
    This, together with \cref{eq:applicationOfPaleyZygmund}, implies
    \begin{equation*}
        \Lambda(H,\|\Delta H\|_F)=\Pr_{r,s}[|F(r,s)|^2\geq \norm{\Delta H}_F^2]\geq \frac{1}{2(9^k+3)} \ge\frac{1}{3\cdot 9^{k}},
    \end{equation*}
    as claimed.
\end{proof}

Putting \cref{lem:spectral-condition,lem:boundonGamma} together, we can conclude that, if $\norm{\Delta H}_F\geq \eps,$ picking $t\in [0,2/\eps]$ uniformly at random, with probability at least 1/3, we have that
\begin{equation*}
    I(t) \le 1 - \frac{1}{12\cdot 9^k}.
\end{equation*}

\paragraph{Lower bounding the probability of measuring identity in Bell sampling.} The probability of measuring identity when performing Bell sampling, i.e. $I(t)$, can be expressed as 
\begin{equation*}
    I(t) = \frac{1}{4^n}\sum_{r,s\in \{0,1\}^n}\cos((\lambda_r-\lambda_s)t) = \frac{1}{4^n}|\text{Tr}(e^{-i\Delta Ht})|^2. 
\end{equation*}
By the definition of the normalized Frobenius norm, we have
\[
\|e^{-i\Delta Ht}-I\|_{F}^2=2-\frac{2}{2^n}\operatorname{Re}\operatorname{Tr}(e^{-i\Delta Ht}).
\]
So we can lower bound $I(t)$ as:
\begin{equation}\label{eq:identity-probability-frobenius}
    I(t) =\frac{1}{4^n}|\text{Tr}(e^{-i\Delta Ht})|^2 \ge (\frac{1}{2^n}\operatorname{Re}\operatorname{Tr}(e^{-i\Delta Ht}))^2 = (1-\frac{1}{2}\|e^{-i\Delta Ht}-I\|_{F}^2)^2 \ge 1-\|e^{-i\Delta Ht}-I\|_F^2.
\end{equation}

Now, we can upper bound $\|e^{-i\Delta Ht}-I\|_F$ using the following simple lemma. 

\begin{lemma}\label{lem:duhamel}
Let $\Delta H$ be an $n$-qubit Hamiltonian and $t>0$. Then, 
\begin{equation*}
    \norm{e^{-it\Delta H}-I^{\otimes n}}_F\leq t\norm{\Delta H}_F.
\end{equation*}
\end{lemma}

\begin{proof}
By the fundamental theorem of calculus we have that 
\begin{equation*}
    \int_0^t-i\Delta He^{-i\Delta H s}ds=e^{-i\Delta Ht}-I^{\otimes n}.
\end{equation*}

Finally, by the fact that the normalized Frobenius norm is invariant under unitaries transformations, we have that $\norm{\Delta He^{-i\Delta H s}}_F=\norm{\Delta H}_F,$ so 
\begin{equation*}
    \norm{e^{-i\Delta Ht}-I^{\otimes n}}_F=\norm{\int_0^t-i\Delta He^{-i\Delta H s}ds}_F\leq \int_0^t\norm{\Delta H}_Fds=t\norm{\Delta H}_F.
\end{equation*}
\end{proof}

Hence, using \Cref{eq:identity-probability-frobenius,lem:duhamel}, for every $t>0$, the following holds:
\begin{equation}\label{lem:lowerboundI(t)}
    I(t) \ge 1-t^2 \|\Delta H\|_F^2.
\end{equation}

\paragraph{The algorithm.} Now, we are ready to present our Hamiltonian certification algorithm, \Cref{algo:tolerantdiagonalcertification}, and analyze its guarantees. 

\begin{algorithm}[h!]
\caption{Tolerant certification of $k$-local Hamiltonians}\label{algo:tolerantdiagonalcertification}
\begin{algorithmic}[1]
\Require Parameters $\eps>0$, $k\in \N$, $H_0$, and time evolution access to $H$ 
\State Set $N$ to be a large enough constant.
\For{$i\in [N]$}
\State Sample $t\in [0,2/\eps]$ uniformly at random.
\State Implement unitary $V$ from \Cref{theo:trotterization}, with $\eps_{\operatorname{Trott}}=\tfrac{1}{384\cdot 9^k}.$
\State Use the algorithm of \cref{lem:memorylessPaulisampling} to obtain $|v_{I^{\otimes n}}'|^2$, that, with probability $\geq 0.95/N$, is an $\tfrac{1}{192\cdot 9^k}$-estimate of $|v_{I^{\otimes n}}|^2$. 
\If{$|v_{I^{\otimes n}}'|^2 \leq 1 - \tfrac{7}{96\cdot 9^k }$}
  \State \Return ``FAR'' 
\EndIf
\EndFor
\State \Return ``CLOSE''
\end{algorithmic}

\end{algorithm}

\begin{theorem}[Tolerant certification of $k$-local Hamiltonians]\label{thm:klocalcertificationdiagonal}
    Let $H$ and $H_0$ be $k$-local $n$-qubit Hamiltonians, where $H_0$ is known and $H$ can be accessed via its time evolution operator. Let $\eps,\delta>0.$ Let $C_{\op}\geq 1$ be such that $\norm{H_0}_{\op},\norm{H}_{\op}\leq C_{\op}$. There is an algorithm that uses $O(9^{2k}\log(1/\delta)/\eps)$ total evolution time to test, with success probability $\ge 1-\delta,$ whether $\|H-H_0\|\le \eps/(8\cdot3^k)$ or $\norm{H-H_0}_{F} \geq \eps.$

    Furthermore, the algorithm uses no ancilla qubits, it makes $O(9^{2k}\log(1/\delta))$ experiments, and it makes $O(3^{5k}(C_{\op}/\eps)^{3/2}\log(1/\delta))$ time evolution queries. The algorithm is robust to at least $\Omega(9^{-k})$ SPAM errors, and the classical post-processing time is $O(9^{2k}\log(1/\delta))$.
\end{theorem}
\begin{proof}
    We will show that \cref{algo:tolerantdiagonalcertification} solves the problem for $\delta=0.1,$ and the statement for general $\delta$ follows via a standard majority vote argument. First, we note that by the Trotterization of \cref{theo:trotterization}, for $U=e^{-it\Delta H}$,
\begin{align*}
|v_{I^{\otimes n}}-u_{I^{\otimes n}}|=\frac{1}{2^n}\Big|\Tr(I^{\otimes n}[U-V])\Big|\le \|U-V\|_{\op}\le \eps_{\operatorname{Trott}}\, ,
\end{align*}
where $V$ is defined in \cref{theo:trotterization}.
In particular, as $|u_{I^{\otimes n}}|,|v_{I^{\otimes n}}|\leq 1,$ we have that 
\begin{align*}
||v_{I^{\otimes n}}|^2-|u_{I^{\otimes n}}|^2|=||v_{I^{\otimes n}}|+|u_{I^{\otimes n}}||\cdot||v_{I^{\otimes n}}|-|u_{I^{\otimes n}}||\le 2\eps_{\operatorname{Trott}}\ .
\end{align*}
Now, we note that all the estimates $|v'_{I^\otimes n}|^2$ of $|v_{I^{\otimes n}}|^2$ coming from \cref{lem:memorylessPaulisampling} are a $1/(192\cdot9^k)$-accurate estimate with probability $\geq 0.95$. Hence, taking the $1/(384\cdot9^k)$ Trotter error into account, we have that, with probability $\geq 0.95,$
\begin{equation}
    ||v'_{I^{\otimes n}}|^2-I(t)|\leq \frac{1}{96\cdot9^k},\label{eq:auxdiagonal}
\end{equation}
where we used that $I(t)=|u_{I^{\otimes n}}|^2.$ From this estimate, we show correctness and then perform a complexity analysis.

\textbf{The far case.} Suppose that $\|\Delta H\|_F \ge \epsilon$. By \cref{lem:spectral-condition,lem:boundonGamma}, with probability $\geq 0.95,$ at least for one of the sampled $t\in [0,2/\eps]$, we have $I(t)\leq 1-1/(12\cdot 9^k)$. Hence, by \cref{eq:auxdiagonal} we have that 
$$|v'_{I^{\otimes n}}|^2 < 1-\frac{1}{12\cdot 9^k}+\frac{1}{96\cdot 9^k}\leq 1-\frac{7}{96\cdot 9^k}.$$
Thus, in this case \cref{algo:tolerantdiagonalcertification} outputs FAR with probability $\geq 0.9$, as desired.

\textbf{The close case.} Suppose that $\|\Delta H\|_F\le \epsilon/(8\cdot 3^k)$. Then, by the lower bound of measuring identity in Bell sampling (see \cref{lem:lowerboundI(t)}) and \cref{eq:auxdiagonal}, we have
$$|v'_{I^{\otimes n}}|^2 > 1-t^2\|\Delta H\|_F^2-\frac{1}{96\cdot 9^k} \geq 1-\frac{4}{\eps^2}\cdot\frac{\eps^2}{64\cdot 9^k}-\frac{1}{96\cdot 9^k} \geq 1-\frac{7}{96\cdot 9^k}.$$
Thus, in this case \cref{algo:tolerantdiagonalcertification} outputs CLOSE with probability $\geq 0.95\geq 0.9$, as desired.

\textbf{Complexity analysis.}
The number of experiments of every iteration of the FOR loop of \cref{algo:tolerantdiagonalcertification}, thanks to \cref{lem:memorylessPaulisampling}, is $O(9^{2k})$. Thus, the total number of experiments is is $O(9^{2k})$, and as each experiments evolves the Hamiltonian for time at most $2/\eps$, we have that the total time evolution is $O( 9^{2k}/\eps)$. By virtue of \cref{theo:trotterization}, for every experiment we need to make $O(3^k(C_{\op}/\eps)^{3/2})$ queries to the time evolution operator, so in total we make $O(3^{5k}(C_{\op}/\eps)^{3/2})$. The robustness to SPAM errors and the classical post-processing time follows from \cref{lem:memorylessPaulisampling}.
\end{proof}

\section{Learning and certifying Gibbs states}
\subsection{Learning Gibbs states}
In this section we propose a fully-sample-efficient protocol to learn Gibbs states, i.e., an algorithm whose sample-complexity is at most polynomial in all relevant parameters. We follow a strategy that reduces learning to hypothesis selection \cite{Yatracos1985Jun,buadescu2021improved}.
First, we show that the following set is an $\eps$-covering net for the set of Gibbs states coming from a $k$-local Hamiltonian with bounded Pauli coefficients:
\begin{equation}
    \mathcal S_{\eps,k,n,\beta}=\left\{ e^{-\beta H}/\Tr[e^{-\beta H}]\ :\ H\in \mathcal H_{\eps,k,n,\beta} \right\},
\end{equation}
where 
\begin{equation*}
    \mathcal H_{\eps,k,n,\beta}=\left\{H: H=\sum_{|P|\leq k} h_{P}P, \; h_{P}\in \eta\mathbb{Z}\cap [-1,1]\right\}
\end{equation*}
and $\eta=\eta_{\eps,k,n,\beta} = \eps/(200\beta n^k)$.
\begin{lemma}\label{lem:net}
    Let $H$ be a $k$-local Hamiltonian acting on $n$ qubits with $|h_P|\leq 1$ for every $P\in\{I,X,Y,Z\}^{\otimes n}$. Then, there exists $\rho\in \mathcal S_{\eps,k,n,\beta}$ such that $\norm{\rho(\beta)-\rho}_{\tr}\leq \eps.$
\end{lemma}
\begin{proof}
    Given $P\in\{I,X,Y,Z\}^{\otimes n}$, let $h_P'$ be the element of $\eta\mathbb Z\cap[-1,1]$ that is closest to $h_P$. Let $H'=\sum h_P'P.$ Then, $\rho=e^{-\beta H'}/\Tr[e^{-\beta H'}]$ belongs to $\mathcal S_{\eps,k,n,\beta}.$ Also, by \cref{lem:dSKLGibbsStates} we have that 
    \begin{align*}
    \norm{\rho(\beta)-\rho}_{\tr}=200\beta n^k\max_{|P|\leq k}|h_P-h_P'|\leq 200\beta n^k\eta = \eps\, ,
    \end{align*}
    where we used the choice of $\eta$ in the last step.
\end{proof}

\noindent Next, we introduce some observables whose expected value will allow us to determine which element of the net is closest to the unknown state. Note that $$|\mathcal H_{\eps,k,n,\beta}|=|\mathcal S_{\eps,k,n,\beta}|=(2/\eta)^{O(n^k)}=(n^{k}\beta/\eps)^{O(n^k)},$$ so we can index the elements of both sets with elements of $[(n^{k}\beta/\eps)^{O(n^k)}]$. For any two indices $i,j\in [(n^{k}\beta/\eps)^{O(n^k)}],$ we define the observable $\Delta H_{i,j}=H_i-H_j.$ First, we bound the number of copies needed to estimate the expected values of all the observables $\Delta H_{i,j}$ in an unknown state $\rho$.

\begin{lemma}\label{lem:observablesfornetlearning}
    Let $\rho$ be an $n$-qubit state, and let $\eps',\tilde \eps,\delta>0.$ Then, with $O(3^kn^{2k}k\log(n/\delta)/\tilde \eps^2)$ single copies of $\rho$ one can obtain estimates $ \Delta H_{i,j,\rho}'$ such that, with probability $\geq 1-\delta,$
    $$|\Delta H_{i,j,\rho}'-\Tr[\rho\Delta H_{i,j}]|\leq \tilde\eps$$ holds simultaneously for every pair of Hamiltonians $H_i,H_j$ belonging to $\mathcal{H}_{ \eps',n,k,\beta}$. The classical post-processing time is $(n^{k}\beta/\eps')^{O(n^k)}/\widetilde\eps^2.$
\end{lemma}
\begin{proof}
    By the classical shadow estimation protocol of \cref{theo:Shadows}, with $O(3^kn^{2k}k\log(n/\delta)/\tilde \eps^2)$ many copies of $\rho$ one can obtain estimates $\rho_P'$ such that, with probability $\geq 1-\delta$, satisfy $$|2^n\rho_P'-2^n\rho_P|= \frac{\widetilde \eps}{200n^k}$$ for every $|P|\leq k.$ We define $\Delta H_{i,j,\rho}'=\sum_{|P|\leq k}((h_i)_P-(h_j)_P)2^n\rho_P'.$ Then, by Plancherel's identity and \cref{eq:numberofklocalPaulis}, we get
    \begin{align*}
        |\Delta H_{i,j,\rho}'-\Tr[\Delta H_{i,j}\rho]|\leq \sum_{|P|\leq k}|(h_i)_P-(h_j)_P|\cdot 2^n|\rho_P-\rho_P'|\leq 200n^k \max_P|2^n\rho_P'-2^n\rho_P|=\tilde\eps.
    \end{align*}
    The classical post-processing time bound to obtain the estimates $\rho'_P$ is $O(3^kn^{3k}k\log(n)/\tilde \eps^2)$, coming from \cref{theo:Shadows}. Once we have the estimates $\rho'_P,$ by \cref{eq:numberofklocalPaulis}, it takes $O(n^k)$ time to compute each $\Delta H'_{i,j}$. Hence, the total post-processing time is $$O((3^kn^{3k}k\log(n)/\tilde \eps^2)+n^k|\mathcal H_{\eps,k,n,\beta}|^2)=O((3^kn^{3k}k\log(n)/\tilde \eps^2))+n^k(n^{k}\beta/\eps')^{O(n^k)}=(n^{k}\beta/\eps')^{O(n^k)}/\widetilde\eps^2 \, .$$ 
\end{proof}

Now, we are ready to present our Gibbs state learning protocol.

\begin{algorithm}[H]
\caption{Gibbs state learning}\label{alg:quantum_data3}

\begin{algorithmic}[1]
\Require $\delta,\eps\in (0,1)$; $O(3^kn^{2k}k\log(n)(\max\{\beta,1\})^2/ \eps^4)$ single copies of $\rho$. Set $\eps'=\frac{\eps^2}{100\max\{\beta,1\}n^k}$.
\State Obtain $(\eps^2/(\max\{\beta,1\}))$-estimates $\Delta H'_{i,j,\rho}$ of $\Tr(\Delta H_{i,j}\rho)$ with probability $\ge 1-\delta$, for pairs $H_{i}$, $H_j$ belonging to $\mathcal H_{\eps',n,k,\beta}$ via the protocol of \cref{lem:observablesfornetlearning}. 
\State Output $\rho'\in \mathcal{S}_{\eps,n.k,\beta}$, where $$\rho'=\text{argmin}_{\tau\in\mathcal{S}_{\eps',n,k,\beta}}\{\max_{i,j}\{|\Delta H_{i,j,\rho}'-\Tr[\Delta H_{i,j}\tau]|\}\}.$$
\end{algorithmic}
\end{algorithm}

\begin{theorem}\label{theo:GibbsStateLearning}
    Let $\rho$ be the Gibbs state at inverse temperature $\beta$ of an $n$-qubit and $k$-local Hamiltonian $H$ with $|h_P|\leq 1$ for every $P.$ Let $\delta,\eps\in (0,1)$. Then, from $O(3^kn^{2k}k\log(n/\delta)(\max\{\beta,1\})^2/ \eps^4)$ single copies of $\rho$, \Cref{alg:quantum_data3} obtains $\rho'\in\mathcal S_{\eps',n,k,\beta}$ such that $\norm{\rho'-\rho}_{\tr}\leq \eps$ with probability $\ge 1-\delta$. The classical post-processing time of the protocol is $(n^{k}\max\{\beta,1\}/\eps)^{O(n^k)}.$
\end{theorem}
\begin{proof}

    We first show correctness, and then perform a complexity analysis. 

    \textbf{Correctness analysis.}  By \cref{lem:net}, there is $\rho''\in\mathcal{S}_{\eps',n,k,\beta}$ such that $\norm{\rho-\rho''}_{\tr}\leq  \frac{\eps^2}{100\max\{\beta,1\}n^k}\le \eps$. In particular, 
    \begin{align}
        \max_{i,j}\{|\Delta H_{i,j,\rho}'-\Tr[\Delta H_{i,j}\rho'']|\}&= \max_{i,j}\{|(\Delta H_{i,j,\rho}'-\Tr[\Delta H_{i,j}\rho])-\Tr[\Delta H_{i,j}(\rho''-\rho)]|\}\nonumber\\
        &\leq \frac{\eps^2}{\max\{\beta,1\}}+\max_{i,j}\norm{\Delta H_{i,j}}_{\op}\norm{\rho''-\rho}_{\tr}
        \\
        &\leq \frac{\eps^2}{\max\{\beta,1\}}+200n^k\cdot \frac{\eps^2}{100\max\{\beta,1\}n^k}\leq 3\frac{\eps^2}{\beta},\label{eq:auxilary}
    \end{align}
    where in the second line we have used the guarantees of \cref{theo:Shadows}, and in the third line we have used that $\norm{\Delta H_{i,j}}_{\op}\leq 200n^k$ because of \cref{eq:numberofklocalPaulis}. Thus, by definition of $\rho'$, we also have 
    \begin{equation}\label{eq:auxiliary2}
        \max_{i,j}\{|\Delta H_{i,j,\rho}'-\Tr[\Delta H_{i,j}\rho']|\}\leq 3\frac{\eps^2}{\beta}.
    \end{equation}
    
\noindent Now, we are ready to upper bound the trace distance between $\rho'$ and $\rho$. By the triangle inequality we have that 
    \begin{align*}
        \norm{\rho-\rho'}_{\tr}\leq \norm{\rho-\rho''}_{\tr}+\norm{\rho'-\rho''}_{\tr}\leq \eps+\norm{\rho'-\rho''}_{\tr}.
    \end{align*}
    By \cref{lem:dSKLGibbsStates}, Equation \eqref{eq:trnormforGibbs0}, we further have that 
    \begin{align*}
        \norm{\rho'-\rho''}_{\tr}&\leq \sqrt{2\beta \Tr[\Delta H_{l_1,l_0}(\rho'-\rho'')]},
    \end{align*}
    where $l_0$, resp. $l_1$, is the label of the Hamiltonian $H_{l_0}$, resp. $H_{l_1}$, corresponding to state $\rho'$, resp. $\rho''$. Next, we apply \cref{eq:auxilary} and \cref{eq:auxiliary2} and get 
    \begin{equation*}
        \norm{\rho-\rho'}_{\tr}\leq \eps+\sqrt{4\times 3\eps^2}\le 5\eps.
    \end{equation*}
The bound claimed in the statement of the theorem follows up to constant rescaling.
    
    \textbf{Complexity analysis.} The complexities follow from applying \cref{lem:observablesfornetlearning} with $$\eps'=\eps^2/(100\max\{\beta,1\}n^k)\qquad \text{ and } \qquad \tilde\eps=\eps^2/\max\{\beta,1\}.$$
\end{proof}

\subsection{Certifying Gibbs states}
In this section we propose a fully-efficient protocol to certify Gibbs states, i.e., an algorithm whose sample-complexity and time-complexity are both at most polynomial in all relevant parameters. 
\begin{theorem}\label{theo:certifyingGibbsStates}
    Let $\rho$ and $\rho_0$ be the Gibbs states at inverse temperature $\beta$ of $n$-qubit and $k$-local Hamiltonians $H$ and $H_0$ with $|h_P|,|(h_0)_{P}|\leq 1$ for every $P,$ respectively. Assume that $H_0$ is known. Let $\delta,\eps\in (0,1)$. Then, \Cref{alg:quantum_data4} decides, with success probability $\geq 1-\delta$, whether $\norm{\rho-\rho_0}_{\tr}\leq \eps^2/(400\beta n^k)$ or $\norm{\rho-\rho_0}_{\tr}\geq 2\eps$ with $$O
    \left(\frac{\beta^2n^{2k}3^kk\log(n/\delta)}{\eps^4}\right)$$  single copies of $\rho$ and $\rho_0$. Moreover, the protocol only requires Pauli measurements, and a classical post-processing time of order $O
    \left(\beta^2n^{3k}3^kk\log(n/\delta)/\eps^4\right)$. The same conclusion holds if $\rho$ and $\rho_0$ are both unknown and we are given copy access to both.\footnote{As for certain regime it happens that $\eps^2/(400\beta n^k)\geq  2\eps$, it may seem that the \emph{far} and \emph{close} can overlap, and thus that the testing task is not well-defined. However, this is not the case, because for that regime of parameters the \emph{far} hypothesis cannot occur. Indeed, by \cref{lem:dSKLGibbsStates} we have that $\norm{\rho(\beta)-\rho_0(\beta)}_{\tr}\leq \sqrt{2\beta \norm{H-H_0}_{\op}\norm{\rho(\beta)-\rho(\beta)}_{\tr}}$. Then, as $\norm{H-H_0}_{\tr}\leq 200 n^k$, because $|h_P|,|(h_0)_P|\leq 1$, we have that $\norm{\rho(\beta)-\rho_0(\beta)}_{\tr}\leq  400\beta n^k$. For the parameters such that $\eps^2/(400\beta n^k)\geq  2\eps$ we then have that $\norm{\rho(\beta)-\rho_0(\beta)}_{\tr}\leq \eps/2.$}
\end{theorem}

\noindent Since the situation where $\rho$ and $\rho_0$ are both unknown is strictly harder than the case of a known $\rho_0$, we only treat the former. 

\begin{algorithm}[H]
\caption{Gibbs state certification}\label{alg:quantum_data4}
\begin{algorithmic}[1]
\Require $O
    \left(\frac{\beta^2n^{2k}3^kk\log(n/\delta)}{\eps^4}\right)$ single copies of $\rho$, $\delta,\eps\in (0,1)$.
\State Obtain estimates $\rho'_P$ and $(\rho_0)'_P$ such that, with probability $\ge 1-\delta$ via the classical shadow tomography protocol of \cref{theo:Shadows}, such that 
\begin{equation}\label{eq:shadowssucceed}
        2^n|\rho_P-\rho_P'|,2^n|(\rho_0)_P-(\rho_0)_P'|\leq \frac{\eps^2}{800\beta n^k},
    \end{equation} 
    for every $|P|\le k$.
\If{there is $|P|\le k$ such that $2^n|\rho_P'-(\rho_0)_{P}'|\geq 3\eps^2/(400\beta n^k)$} \State output ``FAR''.
\Else \State output ``CLOSE''.
\EndIf
\end{algorithmic}
\end{algorithm}

\begin{proof}
    We first show correctness, and then perform a complexity analysis.

    \textbf{Correctness analysis.} Assume that \cref{eq:shadowssucceed} holds. If $\norm{\rho-\rho_0}_{\tr}\leq \eps^2/(400\beta n^k),$ then $$2^n|\rho_P'-(\rho_0)'_P|\leq 2^n|\rho_P'-\rho_P|+2^n|\rho_P-(\rho_0)_P|+2^n|(\rho_0)_P-(\rho_0)_P'|\leq 2\frac{\eps^2}{400\beta n^k},$$
    so we output ``CLOSE'', as desired. 

    On the other hand, assume that $\norm{\rho-\rho_0}_{\tr}\geq 2\eps.$ Then, by \cref{lem:dSKLGibbsStates}
    \begin{align*}
        4\eps^2\leq 400\beta n^k\max_{|P|\leq k}2^n|\rho_P-(\rho_0)_P|.
    \end{align*}
    Now, by \cref{eq:shadowssucceed} we have that 
    \begin{align*}
        3\eps^2\leq 400\beta n^k\max_{|P|\leq k}2^n|\rho_P'-(\rho_0)_P'|.
    \end{align*}
    Hence, there is $|P|\leq k$ such that $2^n|\rho_P'-(\rho_0)_P'|\geq 3\eps^2/(400\beta n^k)$, as desired.

    \textbf{Complexity analysis.} The complexity analysis follows from applying the classical shadow tomography protocol of \cref{theo:Shadows} with error parameter $\eps^2/(800\beta n^k).$
\end{proof}

\section*{Acknowledgments}
J.L. and M.S. thank Savar D. Sinha and Yu Tong for helpful discussions. A.B. was supported by the ANR project PraQPV, grant number ANR-24-CE47-3023. F.E.G. was supported by the European union’s Horizon 2020 research and innovation programme under the Marie Sk lodowska-Curie grant agreement no. 945045, and by the NWO Gravitation project NETWORKS under grant no. 024.002.003. C.R. is supported by France 2030 under the French National Research Agency award number ``ANR-22-PNCQ-0002''.

\bibliographystyle{alphaurl}
\bibliography{Bibliography}
\end{document}